\documentclass[
reprint,
superscriptaddress,
nofootinbib,
amsmath,amssymb,
aps,
pra,
floatfix,longbibliography
]{revtex4-2}
\usepackage{graphicx}% Include figure files
\usepackage{dcolumn}% Align table columns on decimal point
\usepackage[hypertexnames, breaklinks=true, bookmarksnumbered=true, bookmarksopen=true, colorlinks=true, linktocpage=true, citecolor=magenta, urlcolor=magenta, linkcolor=magenta]{hyperref}% add hypertext capabilities
\usepackage{amsmath,amssymb,amsfonts,amsthm}
\usepackage{mathtools} %interesting stuff
\usepackage[T1]{fontenc}
\usepackage[latin9]{inputenc}
\usepackage{txfonts}
\usepackage{bbm} %Enables correct fonts for the fields like C,R etc.
\usepackage{bm} %Bold Math
\usepackage{mathrsfs} %For nice looking curly CAPS
\usepackage{xcolor}
\usepackage{natbib}
\allowdisplaybreaks

\DeclareMathOperator{\tr}{Tr}

\newtheorem{theorem}{Theorem}
\newtheorem{lemma}{\emph{Lemma}}

\newcommand{\bra}[1]{\langle #1|}
\newcommand{\ket}[1]{|#1 \rangle}
\newcommand{\braket}[2]{ \langle #1| #2 \rangle}
\newcommand{\ketbra}[2]{|#1 \rangle\! \langle#2|}
\newcommand{\norm}[1]{ \lVert #1  \rVert}

\begin{document}

\title{Real quantum operations and state transformations}
\author{Tulja Varun Kondra}
\email{t.kondra@cent.uw.edu.pl}
\affiliation{Centre for Quantum Optical Technologies, Centre of New Technologies,
University of Warsaw, Banacha 2c, 02-097 Warsaw, Poland}

\author{Chandan Datta}
\affiliation{Centre for Quantum Optical Technologies, Centre of New Technologies,
University of Warsaw, Banacha 2c, 02-097 Warsaw, Poland}
\affiliation{Institute for Theoretical Physics III, Heinrich Heine University D\"{u}sseldorf, Universit\"{a}tsstra{\ss}e 1, D-40225 D\"{u}sseldorf, Germany}

\author{Alexander Streltsov}

\affiliation{Centre for Quantum Optical Technologies, Centre of New Technologies,
University of Warsaw, Banacha 2c, 02-097 Warsaw, Poland}

\begin{abstract}
Resource theory of imaginarity provides a useful framework to understand the role of complex numbers, which are essential in the formulation of quantum mechanics, in a mathematically rigorous way. In the first part of this article, we study the properties of ``real'' (quantum) operations both in single-party and bipartite settings. As a consequence, we provide necessary and sufficient conditions for state transformations under real operations and show the 
existence of ``real entanglement'' monotones. In the second part of this article, we focus on the problem of single copy state transformation via real quantum operations. When starting from pure initial states, we completely solve this problem by finding an analytical expression for the optimal fidelity of transformation, for a given probability of transformation and vice versa. Moreover, for state transformations involving arbitrary initial states and pure final states, we provide a semidefinite program to compute the optimal achievable fidelity, for a given probability of transformation.
\end{abstract}
\maketitle
\section{Introduction}
In the early 18th-century complex numbers were introduced, soon to be realised that not only it is an indispensable part of mathematics, but also very useful in natural sciences. In fact, in physics, it has been widely used in the study of waves. It is no surprise that complex numbers played a crucial role in quantum physics as well. According to the postulates of quantum physics (finite dimensional), we describe the state of a physical system as a positive semidefinite operator with unit trace, acting on a complex Hilbert space. Naturally, the question about the necessity of complex Hilbert space arises. To put it another way, instead of complex Hilbert spaces, can we restrict ourselves to use only real Hilbert spaces to explain all the phenomena in quantum physics. Very recently, this question has been answered negatively\cite{Wu_PRL,renou,PhysRevLett.128.040403,PhysRevLett.128.040402}. 

Quantum resource theories \cite{Chitambar_2019} on the other hand, provide a unified approach for studying various quantum phenomena and their applications in quantum information protocols. Although the study of entangled state transformations under local operations and classical communications \cite{Bennett_1993,Bennett_1996,Vedral_97,Nielsen_1999,Horodecki_2009} has started in the 90's, the term \emph{resource theory} was originally introduced in a paper on resource theory of asymmetry by Gour and Spekkens in 2008 \cite{Gour_2008}. Since then, many other resource theories have been developed, as example, coherence \cite{ RevModPhys.89.041003}, purity \cite{ Streltsov_2018}, thermodynamics \cite{Goold_2016, Lostaglio_2019}, and stabilizer quantum computation \cite{Veitch_2014}. Very recently, operational resource theory of imaginarity \cite{Wu_PRL, IM2} has been introduced, capturing the effort to create and manipulate complex quantum states. The study of resource theory of imaginarity not only has a fundamental importance but also provides an operational meaning\cite{IM2}. 

In any resource theory, one of the main goals is to study the resource transformation under free operations. To be precise, given two quantum states $\rho$ and $\sigma$, whether we can transform $\rho$ to $\sigma$ under free operations specified by the theory. This idea of transforming a quantum state into a more desired state has also been the central focus of quantum control theory \cite{JOHANSSON20131234,RAZA2021107541}. In recent literature state transformations in different scenarios have been studied, namely deterministic \cite{Nielsen_1999,WinterPhysRevLett.116.120404,PhysRevA.91.052120}, stochastic \cite{vidal_prob,vidal_purestates,coherence_pure,regula2021probabilistic}, and approximate \cite{Vidal_approx,RegulaPhysRevA.101.062315}.  
In the formulation of the resource theory of imaginarity, one identifies the real density matrices as free states and free operations as the completely positive trace preserving (CPTP) maps which can be represented by real Kraus operators \cite{Wu_PRL}. Note that, similar to the resource theory of coherence \cite{plennio_coh}, the resource theory of imaginarity is a basis dependent theory. When we say an operator is ``real'', we mean it has only real matrix elements in the chosen (fixed) basis. Once we identify the set of free states and free operations, we can look at the problem of deterministic state transformations i.e, which state transformations are possible via free operations. In \cite{Wu_PRL, IM2}, authors consider deterministic conversion of qubit states under real operations. Here, we extend these results by presenting necessary and sufficient conditions for arbitrary state transformations under real operations. In the case, when deterministic transformation between two states is not possible, one can look at the possibility of a probabilistic transformation via stochastic free operations. In our case, the stochastic free operations correspond to CP trace non-increasing maps with real Kraus operators. The central goal of the problem of probabilistic (or stochastic) transformations is to find out the optimal probability with which a given transformation is possible.  Previously, in \cite{Wu_PRL, IM2}, the optimal probability for the conversion between pure states under real operations has been derived. Here, we extend these results by finding the optimal probability of transforming a pure state into an arbitrary state. A more general question about state transformation allows for a small error in the final state, and the goal is to find the optimal probability of this ``approximate'' transformation. This is the most general form of single copy transformation called ``stochastic-approximate'' transformation \cite{Varun_stochastic}. Recently this scenario has been explored for general resource theories \cite{Varun_stochastic,Regula_distillation}. For the case when the initial state is pure, we provide an analytical expression for the optimal probability of transformation, for a given allowed error in the transformation. We also provide an analytical expression for the optimal achievable fidelity, given a probability of the transformation. Furthermore, we provide a semidefinite program to compute the optimal fidelity of transforming an arbitrary state into a pure state, via real operations, for a given probability of transformation. 

\section{Properties of real quantum operations}
\subsection{Real quantum operations in single-party setting}
First we formally introduce the resource theory of imaginarity \cite{Wu_PRL,IM2}. Given a reference basis $\{\ket{i}\}$, the free states in this theory are those states that have real density matrices. Specifically a state $\rho$ is called a free state when $\bra{i}\rho\ket{j}\in \mathbb{R}$ for all $i$ and $j$. The free operations are those which can be characterised by real Kraus operators. A operation $\Lambda$ characterised by Kraus operators $\{K_m\}$ is free if $\bra{i}K_m\ket{j}\in \mathbb{R}$ for all $i$, $j$ and $m$. In this section, we will explore state transformations under covariant operations. Here, covariant means, covariant with respect to transpose. A CP map ($\mathcal{E}$) is said to be ``$\rho$-covariant'' iff $(\mathcal{E}(\rho))^T=\mathcal{E}(\rho^T)$. Similarly, a CP map ($\Lambda$) is said to be ``covariant'' iff $(\Lambda(\rho))^T=\Lambda(\rho^T)$ for all $\rho$. But before presenting the results we introduce the following lemma which will be used later. Note that, in this letter by CP maps we mean completely positive trace non-increasing maps. We will now use these definitions to prove some interesting properties of real CP maps.

\begin{lemma} \label{kraus}
Let $\mathcal{E}_1$ and $\mathcal{E}_2$ be two CP maps with Kraus operators $\{P_j\}$ and $\{Q_j\}$ where $j \in \{1,\cdots,n\}$. Then a new CP map $\mathcal{E}$ can be defined as
\begin{equation}
 \mathcal{E}(\cdot) = \frac{1}{2}\mathcal{E}_1(\cdot) + \frac{1}{2}\mathcal{E}_2(\cdot),
\end{equation}
where $\mathcal{E}$ has Kraus operators given by $\left\{\frac{P_j + Q_j}{2}, \frac{i(P_j - Q_j)}{2}\right\}$.
\end{lemma}
\begin{proof}
\begin{eqnarray}
 \mathcal{E}(\cdot) &=& \sum_j \left(\frac{P_j + Q_j}{2}(\cdot) \frac{P^{\dagger}_j + Q^{\dagger}_j}{2} + \frac{i(P_j - Q_j)}{2}(\cdot) \frac{-i(P^{\dagger}_j - Q^{\dagger}_j)}{2}\right) \nonumber\\
 &=& \sum_j\left( \frac{1}{2}P_j (\cdot) P^{\dagger}_j +\frac{1}{2}Q_j (\cdot) Q^{\dagger}_j\right) \nonumber \\
 &=& \frac{1}{2}\mathcal{E}_1(\cdot) + \frac{1}{2}\mathcal{E}_2(\cdot), \nonumber
\end{eqnarray}
which completes the proof.
\end{proof}
Using the above lemma, we now prove that, all covariant CP maps can be represented by real Kraus operators. 

\begin{theorem}\label{covariant_real}
All covariant CP maps can be expressed with real Kraus operators.
\end{theorem}
\begin{proof}
Let a covariant map $\Lambda$ have a Kraus representation given by $\{L_j\}$, therefore the following equations hold for all $\rho$
\begin{eqnarray}
    &&\sum_j L_j \rho L^\dagger_j = \sigma, \label{1stcov}\\
    &&\sum_j L_j \rho^T L^\dagger_j = \sigma^T. \label{2ndcov}
\end{eqnarray}
Taking transpose of (\ref{1stcov}) and (\ref{2ndcov}), gives the following equations
\begin{eqnarray}
   && \sum_j L^*_j \rho^{T} L^T_j = \sigma^{T},\\ \label{1st}
    && \sum_j L^*_j \rho L^T_j = \sigma. \label{2nd}
\end{eqnarray}
Therefore, $\{L_{j}\}$ and $\{L^*_{j}\}$ are Kraus representations of the same covariant operation $\Lambda$. We can re-write the operation in the following way
\begin{equation}
\Lambda(\cdot) = \frac{1}{2}\Lambda(\cdot) + \frac{1}{2}\Lambda(\cdot)=\frac{1}{2}\sum_j L_j (\cdot) L_j^\dagger+\frac{1}{2}\sum_j L_j^* (\cdot) L_j^T.
\end{equation}
Using Lemma \ref{kraus}, we see that these Kraus operators associated with $\Lambda$ can be expressed as $\left\{\frac{L_j +L^*_j}{\sqrt{2}}, \frac{i(L_j - L^*_j)}{\sqrt{2}}\right\}$. This completes the proof as the Kraus operators are real.
\end{proof}
It is easy to see that any CP map with real Kraus operators is covariant. Using this fact, along with the above theorem, one can say that a CP map is real iff it is covariant. $\rho$-covariant maps correspond to a larger set than covariant maps. But when considering state transformations of $\rho$, we will see that, covariant maps are as powerful as $\rho$-covariant maps.
\begin{theorem}
Let a transformation ($\rho \rightarrow \sigma$) be possible via $\rho$-covariant operations, then the transformation is also possible via covariant operations. \label{connection}
\end{theorem}
\begin{proof}
Let $\Lambda_1$ be a $\rho$-covariant operation which achieves the desired transformation, having Kraus operators $\{K_{j}\}$. Therefore, the following relations hold
\begin{eqnarray}
    &&\sum_i K_j \rho K^\dagger_j = \sigma,\label{1st}\\ 
   && \sum_i K_j \rho^T K^\dagger_j = \sigma^T. \label{2nd}
\end{eqnarray}
Applying Transpose operations on both sides of (\ref{1st}) and (\ref{2nd}) gives the following equations
\begin{eqnarray}
  &&  \sum_j K^*_j \rho^{T} K^T_j = \sigma^{T},\\ 
    &&\sum_j K^*_j \rho K^T_j = \sigma. 
\end{eqnarray}
Therefore, we arrive at a new $\rho$-covariant operation ($\Lambda_2$) which achieves the same desired transformation, with a Kraus representation given by $\{K^*_{j}\}$. Next, we design a covariant operation ($\Lambda_3$), such that
\begin{equation}
\Lambda_3(\cdot) = \frac{1}{2}\Lambda_1(\cdot) + \frac{1}{2}\Lambda_2(\cdot),
\end{equation}
and using Lemma \ref{kraus}, we know that this operation has Kraus operators given by $\left\{\frac{K_j +K^*_j}{\sqrt{2}}, \frac{i(K_j - K^*_j)}{\sqrt{2}}\right\}$ (note that these Kraus operators are real). This completes the proof.
\end{proof}

Next, we show that Theorem \ref{connection} can be used to derive necessary and sufficient conditions for state transformations under real quantum operations. In \cite{QuantumMajorization}, the authors derived necessary and sufficient conditions for the existence of a CPTP map ($\mathcal{E}$), such that
\begin{equation}\label{sets}
\mathcal{E}(\rho_i) = \sigma_i \,\, \forall i \in \{1,\cdots,n\}, \end{equation}
where $\{\rho_i\}$ and $\{\sigma_i\}$ are two sets composed of $n$ density matrices each. It's easy to see that, a transformation $\rho \rightarrow \sigma$ is possible via $\rho$-covariant operations iff there exists a CPTP map tranforming the set $\{\rho,\rho^T\}$ into $\{\sigma,\sigma^{T}\}$. Therefore, the transformation conditions for Eq. (\ref{sets}), along with Theorem \ref{connection}, allow us to give necessary and sufficient conditions for the state transformations under real quantum operations, by substituting $\{\rho, \rho^T\}$ and $\{\sigma, \sigma^T\}$ instead of $\{\rho_i\}$ and $\{\sigma_i\}$.  For the sake of completeness we provide these conditions below. We refer to \cite{QuantumMajorization} for a more detailed discussion.

Let $\{\rho, \sigma\}$ be density matrices acting on $\{\mathcal{H}_{\mathrm{B}}, \mathcal{H}_{\mathrm{C}}\}$ and $\left|\phi_{+}^{\mathrm{AC}}\right\rangle$ denotes the maximally entangled state on $\mathcal{H}_{\mathrm{A}} \otimes \mathcal{H}_{\mathrm{C}}$. Additionally we define $d_A = d_C$, where $d_A, d_B$ and $d_C$ are the dimensions of the Hilbert spaces $\mathcal{H}_{\mathrm{A}}, \mathcal{H}_{\mathrm{B}}$ and $\mathcal{H}_{\mathrm{C}}$ respectively.
\begin{equation}
    \Omega^{ABC} = \frac{1}{2}\omega_1 \otimes \rho \otimes \sigma + \frac{1}{2}\omega_2 \otimes \rho^T \otimes \sigma^T. \nonumber
\end{equation}
Then the following statements are equivalent:

1. There exists a real CPTP map $\mathcal{E}_r$, such that               $\mathcal{E}_r(\rho) = \sigma$.

2. For any two states $\omega_1$ and $\omega_2$, we have 
$$
2^{-H_{\min }(A \mid B)_{\Omega}} \geqslant d_{\mathrm{C}}\left\langle\phi_{+}^{\mathrm{AC}}\left|\Omega^{\mathrm{AC}}\right| \phi_{+}^{\mathrm{AC}}\right\rangle.
$$

3. For any  $\{\omega_1, \omega_2 \} \in \mathcal{B}\left(\mathcal{H}_{\mathrm{A}}\right)$ : 
$$
H_{\min }(A \mid B)_{\Omega} \leqslant H_{\min }(A \mid C)_{\Omega}.
$$

4. $\alpha = 1$, where
\begin{eqnarray}
\alpha \equiv \min \operatorname{Tr}[Z] \nonumber\\
\text { subject to } &\quad I^{\mathrm{A}} \otimes Z \geqslant  X_{1} \otimes \rho + X_{2} \otimes \rho^T \nonumber\\
  & \operatorname{Tr}\left(\sigma^{\mathrm{T}} X_1 + \sigma X_2 \right)=1 ; X_{1}, X_{2} \geqslant 0 .
\end{eqnarray}
Here, $H_{\min }(R \mid A)_{\Omega}$ is the quantum conditional min-entropy of the state $\Omega^{RA}$, defined as
\begin{equation}
H_{\min }(R \mid A)_{\Omega}:=-\log \inf _{X_{A} \geq 0}\left\{\operatorname{tr}\left[X_{A}\right]: \mathbb{I}_{R} \otimes X_{A} \geq \Omega^{R A}\right\}.
\end{equation}

The above statements provides necessary and sufficient conditions to check whether a transformation $\rho\rightarrow\sigma$ is possible under real operations. Note that the optimisation problem in point 4, is a semidefinite program. 

\subsection{Real quantum operations in bipartite setting}
In this section, we study the nature of real operations in spatially separated bipartite scenarios. We introduce a new class of monotones under local real operations and classical communication (LRCC) \cite{IM2,Wu_PRL} which are independent to local operation and classical communication (LOCC) monotones. Let's assume a bipartite state is shared between Alice and Bob, where Alice and Bob share a classical channel and can perform LRCC operations. Note that, for any two real CP maps $\Lambda_1$ and $\Lambda_2$ 
\begin{equation}\label{covariance}
   ( \Lambda_1\otimes\Lambda_2 (\rho^{AB}))^{T_{B}}=\Lambda_1\otimes\Lambda_2((\rho^{AB})^{T_B})\,\,\textrm{for all}\,\,\rho^{AB}.
\end{equation}
Here, $T_{B}$ is the partial transpose on Bob's system. This is easy to see, as $\Lambda_2$ has only real Kraus operators. This implies that for any LRCC map ($\Lambda$),
\begin{equation}\label{LRCC_COV}
    ( \Lambda (\rho^{AB}))^{T_{B}}=\Lambda((\rho^{AB})^{T_B})\,\,\textrm{for all}\,\,\rho^{AB}.
\end{equation}
Using data-processing inequality of trace distance and Eq. (\ref{LRCC_COV}), we get
\begin{align}\label{LRCC_monotones}
    \norm{\rho^{AB}-(\rho^{AB})^{T_B}}_1&\geq  \norm{\Lambda(\rho^{AB})-\Lambda((\rho^{AB})^{T_B}))}_1\nonumber\\
    &\geq \norm{\Lambda(\rho^{AB})-\Lambda(\rho^{AB})^{T_B})}_1
\end{align}
Note that, the above inequality holds for all bipartite states $\rho^{AB}$ and all LRCC operations $\Lambda$. This shows that,
\begin{equation}
    D(\rho^{AB})= \norm{\rho^{AB}-(\rho^{AB})^{T_B}}_1
\end{equation}
is a real entanglement monotone (as it does not increase under LRCC operations). Here, instead of trace distance, one can also choose  any arbitrary distinguishability measure obeying data processing inequality under CPTP maps. An example of such a distinguishability measure would be the infidelity ($1-F$) between $\rho^{AB}$ and $(\rho^{AB})^{T_B}$. In this case, the infidelity ($1-F$) between $\rho^{AB}$ and $(\rho^{AB})^{T_B}$ can act as a LRCC monotone. Note that in Eq. (\ref{LRCC_monotones}) one can replace $T_B$ with $T_A$ (partial transpose on Alice's system) and get a separate set of monotones.

As a special case, for any real bipartite state $\rho^{AB}$ with $\norm{\rho^{AB}-(\rho^{AB})^{T_B}}_1$>0, Eq. (\ref{LRCC_monotones}) implies that there cannot exist an LRCC operation transforming $\ket{00}\bra{00}^{AB}$ into $\rho^{AB}$. Note that here $\rho^{AB}$ can also be a separable state. For example, lets take a real state which is ``separable'': $\eta^{AB}=\frac{1}{4}(\openone+\sigma_y\otimes\sigma_y)$, where $\sigma_y$ is the Pauli-y matrix. It is easy to see that $||\eta^{AB}-(\eta^{AB})^{T_B}||_{1}>0$. This shows that, even though this state is separable, it is ``entangled'' in real quantum theory. Therefore, this shows that invariance under partial transpose is a necessary condition for separable states in ``real'' quantum theory, which was shown in \cite{Chiribella_2022,caves}. 

\section{Single-copy state transformations}
\subsection{Geometric measure of imaginarity}
The geometric imaginarity of a pure state $\ket{\psi}$ can be expressed as \cite{IM2}
\begin{equation}
    \mathscr{I}_g(\ket{\psi})=1-\max_{\ket{\phi}\in \mathscr{R}}|\braket{\phi}{\psi}|^2, \label{gi_pure}
\end{equation}
where the maximization is taken over all the real pure states. The definition can be extended to mixed states by considering the minimization of average imaginarity over all pure state decompositions as follows  
\begin{equation}
    \mathscr{I}_g(\rho)=\min_{\{p_j,\ket{\psi_j}\}}\sum_{j} p_j  \mathscr{I}_g(\ket{\psi_j}), \label{gi_mixed}
\end{equation}
such that $\rho=\sum_j p_j \ketbra{\psi_j}{\psi_j}$. Note that the above definition of geometric imaginarity is equivalent to the following distance-based measure \cite{Alex_Linking}
\begin{equation}
    \mathscr{I}_g(\rho)=1-\max_{\sigma\in \mathcal{R}}F(\rho,\sigma),
\end{equation}
where $\mathcal{R}$ represents set of all real states and $\sqrt{F(\rho,\sigma)} = \tr \sqrt{\sqrt{\rho} \sigma \sqrt{\rho}}$ is the root fidelity. In \cite{IM2}, analytical expression of geometric measure of imaginarity has been derived for pure states. In the following theorem, we extend the previous result to find an analytical expression for the geometric measure of imaginarity for arbitrary states.
\begin{theorem}\label{Thm:geom_meas}
For a quantum state $\rho$, the geometric measure of imaginarity is given by
\begin{equation}
    \mathscr{I}_g(\rho) = \frac{1 - \sqrt{F(\rho,\rho^T)}}{2}, \label{geoimexp}
\end{equation}
\end{theorem}
The proof is given in the Appendix. In the following, we show that geometric measure of imaginarity plays a key role in probabilistic and stochastic-approximate state transformations.

\subsection{Probabilistic transformations}

As described in the introduction, sometimes exact deterministic state transformations are not possible. Then, we consider the probabilistic scenario. Previously, probabilistic transformations between pure states have been studied \cite{Wu_PRL,IM2}. Here in the following theorem, we extend this to a scenario where the target state is mixed.  
\begin{theorem}
The optimal probability with which $\ket{\psi}$ can be transformed into $\rho$ via real operations is given by 
\begin{equation} \label{opt_prob}
    P(\ket{\psi} \rightarrow \rho ) = \min[\frac{\mathscr{I}_g (\ket{\psi})}{\mathscr{I}_g (\rho)},1].
\end{equation}
\end{theorem}
\begin{proof}
From \cite{IM2}, we know that the ratio of geometric measures gives the upper bound for the optimal achievable probability of transforming $\sigma$ to $\rho$
\begin{equation}
     P(\sigma \rightarrow \rho ) \leq \min[\frac{\mathscr{I}_g (\sigma)}{\mathscr{I}_g (\rho)},1].\label{probstrat}
\end{equation}
In \cite{IM2}, it has been shown that this inequality is saturated for probabilistic pure to pure state transformations i.e, the optimal probability is given by the ratio of geometric measures 
\begin{equation}
     P(\ket{\psi} \rightarrow \ket{\phi} ) = \min[\frac{\mathscr{I}_g (\ket{\psi})}{\mathscr{I}_g (\ket{\phi})} ,1].\label{probpure}
\end{equation}
In the following, we prove that Eq. (\ref{probstrat}) is saturated whenever $\sigma$ is a pure state. Let's assume $\sigma$ = $\ket{\psi}\bra{\psi}$.
Recalling the definition of geometric imaginarity, we assume $\{p'_j,\ket{\psi'_j}\}$ as the ensemble which achieves the minimisation in Eq. (\ref{gi_mixed}) for the state $\rho=\sum_j p_j'\ketbra{\psi_j'}{\psi_j'}$. Furthermore, we choose the following purification of $\rho$,
\begin{equation}
    \ket{\rho} = \sum_j \sqrt{p_j'}\ket{\psi_j'} \otimes \ket{j}^{A},
\end{equation}
where $A$ denotes an ancillary system used for purification. It is easy to check that 
\begin{equation}
\mathscr{I}_g (\ket{\rho}) =  \mathscr{I}_g(\rho).
\end{equation}
From Eq. (\ref{probpure}), we find
\begin{equation}
     P(\ket{\psi} \rightarrow \ket{\rho} ) = \min[\frac{\mathscr{I}_g (\ket{\psi})}{\mathscr{I}_g (\ket{\rho})},1]=\min[\frac{\mathscr{I}_g (\ket{\psi})}{\mathscr{I}_g (\rho)},1].\label{prob_pure_mix}
\end{equation}
Here, we probabilistically converted $\ket{\psi}$ to $\ket{\rho}$, which is a bipartite pure state of system (of $\rho$) and ancilla $A$, with probability given in Eq. (\ref{prob_pure_mix}). We will then discard the ancilla to achieve the desired transformation. This completes the proof.
\end{proof}
The above result gives an optimal probability of converting a pure state into a mixed state and it can be easily computed as geometric measure is easy to find. 

\subsection{Stochastic-approximate transformations}
In the above, we consider exact transformations. Here, we consider a more general scenario where we consider both the stochastic and approximate scenarios together. Stochastic-approximate state transformations have been studied previously for general resource theories \cite{Varun_stochastic,regula2021probabilistic,Fang_2019,Regula_distillation}. The probability for stochastic-approximate conversion is defined as the maximum probability of converting a state $\rho$ into $\sigma$ with fidelity at least $f$ \cite{Varun_stochastic}:
\begin{equation}
    P_{f}(\rho\rightarrow\sigma)=\max_{\Lambda}\left\{\tr \Lambda(\rho) : F\left(\frac{\Lambda(\rho)}{\tr[ \Lambda(\rho)]},\sigma\right)\geq f\right\},
\end{equation}
where $\Lambda$ corresponds to the set of all real operations. Similarly, the fidelity for stochastic-approximate conversion provides the maximal fidelity of transforming a state $\rho$ into $\sigma$ with a success probability at least $p$ and can be expressed as \cite{Varun_stochastic}
\begin{equation}
     F_{p}(\rho\rightarrow\sigma)=\max_{\Lambda}\left\{F\left(\frac{\Lambda(\rho)}{\tr[ \Lambda(\rho)]},\sigma\right) : \tr \Lambda(\rho)\geq p \right\}.
\end{equation}
Before, going into the result we introduce some important results, which will be used to find the above quantities.   
\begin{theorem}
\label{thm:Geom_Continuity}
For an arbitrary state $\rho$, consider a set of states $S_{\rho,f}$ such that $F(\rho,\rho')\geq f$ for all $\rho' \in S_{\rho,f}$. The minimal geometric measure of imaginarity in $S_{\rho,f}$ is given by 
\begin{align}
\min_{\rho'\in S_{\rho,f}}\mathscr{I}_g(\rho') & =\sin^{2}\left(\max\left\{ \sin^{-1}\!\sqrt{\mathscr{I}_g(\rho)}-\cos^{-1}\!\sqrt{f},0\right\} \right).
\end{align}
For any pure state $\ket{\psi}$, the maximal geometric measure of imaginarity in $S_{\psi,f}$ is given by 
\begin{align}
    \max_{\rho'\in S_{\psi,f}}\mathscr{I}_g(\rho') & =\sin^{2}\left(\min\left\{ \sin^{-1}\!\sqrt{\mathscr{I}_g(\ket{\psi})}+\cos^{-1}\!\sqrt{f},\frac{\pi}{4}\right\} \right).
\end{align}
\end{theorem}
For a given target fidelity $f$ to achieve a state $\rho$, the above theorem gives the least possible geometric measure among all states $\rho'$ such that the fidelity between $\rho$ and $\rho'$ is at least $f$. Note that, the case of maximum geometric measure is discussed only for pure states and the case for mixed states remains open. In order to prove this theorem, we use the following lemma. We refer to the Appendix for the proof of the following lemma.

\begin{lemma} \label{equal}
For any state $\rho$, there exists a pure state decomposition $\{p_i,\ket{\psi_i}\}$, such that, $\mathscr{I}_g(\ket{\psi_i})=\mathscr{I}_g(\rho)$ for all i.
\end{lemma}
Similar to the two-qubit entanglement theory \cite{vidal_prob,Wei_2003,wooters_entform}, this lemma shows that there exists a pure state decomposition for any state $\rho$, such that all the pure states have same geometric measure as $\rho$. Using this lemma, we proove theorem \ref{thm:Geom_Continuity} in the Appendix. Equipped with this we are now ready to discuss the stochastic-approximate transformations of an initial pure state.

\begin{theorem} \label{thm:PureConversion}
The maximal probability to convert a pure state $\ket{\psi}$ into an arbitrary state $\rho$ via real operations with a fidelity $f$ is given by
\begin{equation}\label{eq:optimal_probability_imaginarity}
P_{f}(\ket{\psi}\rightarrow\rho) = \begin{cases}
1\,\,\,\mathrm{for}\,\,\,m_1\geq0\\
\frac{\mathscr{I}_g(\ket{\psi})}{\sin^{2}\left(\sin^{-1}\sqrt{\mathscr{I}_g(\rho)}-\cos^{-1}\sqrt{f}\right) }\,\,\mathrm{otherwise},
\end{cases}
\end{equation}
where $m_1=\sin^{-1}\sqrt{\mathscr{I}_g(\ket{\psi})}-\sin^{-1}\sqrt{\mathscr{I}_g(\rho)} +\cos^{-1}\sqrt{f}$.
\end{theorem}
Note that, if we want exact transformation, i.e., $f=1$, then we recover the result stated in Eq. (\ref{prob_pure_mix}). The optimal fidelity for a given probability can also be found easily:
\begin{equation}
F_{p}(\ket{\psi}\rightarrow\rho)=\begin{cases}
1\,\,\,\mathrm{for}\,\,\,p\leq\frac{\mathscr{I}_g(\ket{\psi})}{\mathscr{I}_g(\rho)},\\
\cos^{2}\left[\sin^{-1}\!\sqrt{\mathscr{I}_g(\rho)}-\sin^{-1}\!\sqrt{\frac{\mathscr{I}_g(\ket{\psi})}{p}}\right]\,\,\mathrm{otherwise}. \label{eq:optimal_fidelity_imaginarity}
\end{cases}
\end{equation}
The details can be found in the appendix. The above result provides a complete solution for state transformations under real operations, when starting from a pure state. Note that for a special case when $\rho$ is a maximally imaginary state and $p=1$, we get back the result derived in \cite{Varun_stochastic}. 

In the following, we show that when starting from an arbitrary state $\rho_A=\sum_{i,j}c_{ij}\ket{i}\!\bra{j}$, for a given probability of transformation $p$, the optimal achievable fidelity to reach a pure target state $\ket{\psi_B}$, can be computed by a semidefinite programe (SDP). Note that a CP map is real iff the choi matrix associated with it is real \cite{Gour2018}. Let's take a real CP map $\Lambda:A\rightarrow B$ and $\Sigma_{\Lambda}$ is the choi matrix of $\Lambda$, given by
\begin{equation}
    \Sigma_{\Lambda}=\openone\otimes\Lambda\left(\sum_{i,j}\ket{i}\!\bra{j}\otimes\ket{i}\!\bra{j}\right).
\end{equation}
Here, $d$ is the dimension of the Hilbert space of $A$. It is known that (see Eq. (4.2.12) of \cite{https://doi.org/10.48550/arxiv.2011.04672})
\begin{equation}\label{choi}
    \Lambda(\rho_A)=\tr_{A}(\Sigma_{\Lambda}(\rho_A^{T}\otimes \openone_{B})),
\end{equation}
where $T$ represents the transpose in above mentioned basis. Using Eq. (\ref{choi}), for any pure state $\ket{\psi_B}$ 
\begin{eqnarray}\label{inner}
    \bra{\psi_B}\Lambda(\rho_A)\ket{\psi_B}&=\tr(\Sigma_{\Lambda}(\rho_A^{T}\otimes \ket{\psi_B}\!\bra{\psi_B}))
\end{eqnarray}
Note that, $\Lambda'$ can be considered as another real operation which completes the real operation $\Lambda$ (i.e $\Lambda + \Lambda'$ is a real CPTP map).
This leads us to the following theorem.
\begin{theorem}
The optimal achievable fidelity of transforming a quantum state $\rho$ (acting on a Hilbert space $A$) into a pure state $\ket{\psi}\bra{\psi}$ (acting on a Hilbert space $B$) with a probability $p$, is given by the following semidefinite program.
\\Maximise:
\begin{equation}
    \frac{1}{p}\tr (\Sigma_{\Lambda}(\rho^{T}\otimes \ket{\psi}\!\bra{\psi}))\,\, 
\end{equation}
under the constraints,
\begin{eqnarray}\nonumber
    &&\Sigma_{\Lambda}\geq 0,\, 
     \tr_{B}  \Sigma_{\Lambda}\leq\openone,\,
     \frac{1}{p}\tr (\Sigma_{\Lambda}(\rho^T\otimes \openone))=1\,\,\textrm{and}\,\,\Sigma_{\Lambda}^T=\Sigma_{\Lambda}. 
\end{eqnarray}
\end{theorem}
Using this SDP, we can compute the  optimal achievable fidelity for transforming an arbitrary state into a pure state, for a given probability of transformation. 
\section{Conclusions}
We investigated the properties of real quantum operations both in single party and bipartite settings. This allowed us to characterise the state transformations achievable via real quantum operations and showed the existence of ``real entanglement monotones'', which are independent to the entanglement monotones in complex quantum theory. We then studied the problem of stochastic-approximate state conversion via real quantum operations and for the case of pure initial states, we provide a complete solution. In the case when the final state is pure, we provide a SDP to compute the largest achievable fidelity of transformation, given a probability of success.

\section*{Acknowledgements} This work was supported by the
"Quantum Optical Technologies" project, carried out within the
International Research Agendas programme of the Foundation for Polish
Science co-financed by the European Union under the European Regional
Development Fund and the "Quantum Coherence and Entanglement for Quantum
Technology" project, carried out within the First Team programme of the
Foundation for Polish Science co-financed by the European Union under
the European Regional Development Fund. CD acknowledges support from the German Federal
Ministry of Education and Research (BMBF) within the funding program ``quantum technologies -- from basic research to market'' in the joint project QSolid (grant
number 13N16163).
\section*{Appendix}

\subsection{Proof of Theorem \ref{Thm:geom_meas}}\label{Thm:geom_meas_proof}
We note that, for pure states, geometric measure of imaginarity is given by \cite{IM2}
\begin{equation}
    \mathscr{I}_g(\ket{\psi}) =\frac{1-|\braket{\psi^*}{\psi}|}{2} =\frac{1 - \sqrt{F(\ket{\psi}\!\bra{\psi},\ket{\psi^*}\!\bra{\psi^*})}}{2}.\label{geo_im_exp_pure}
\end{equation}
Let $\{p_j,\ket{\psi_j}\}$ be the ensemble of $\rho$ which achieves the minimisation in Eq. (\ref{gi_mixed}). Using the joint concavity property of root-fidelity \cite{wilde_2017}, we find
\begin{eqnarray}
\sqrt{F(\rho,\rho^T)} &\geq& \sum_j p_j \sqrt{F(\ketbra{\psi_j}{\psi_j},\ket{\psi_j^*}\bra{\psi_j^*})}\nonumber\\
    \frac{1 - \sqrt{F(\rho,\rho^T)}}{2} &\leq& \sum_j p_j \frac{1 -\sqrt{F(\ket{\psi_j}\bra{\psi_j},\ket{\psi^*_j}\bra{\psi^*_j})}}{2}\nonumber\\
    &=& \mathscr{I}_g(\rho),
\end{eqnarray}
where we consider the fact that $\sum_j p_j=1$ and in the last line we use the expression for pure states given in Eq. (\ref{geo_im_exp_pure}). 
In the above, we derive a lower bound for geometric imaginarity, given by,
\begin{equation}
    \mathscr{I}_g(\rho) \geq \frac{1 - \sqrt{F(\rho,\rho^T)}}{2}. \label{geo_im_lower}
\end{equation}
Since, we know
\begin{equation}
    \mathscr{I}_g(\rho)=\min_{\{p_j,\ket{\psi_j}\}}\sum_{j} p_j  \mathscr{I}_g(\ket{\psi_j}), 
\end{equation}
with $\sum_j p_j \ket{\psi_j}\!\bra{\psi_j}=\rho$ and for each pure state $\ket{\psi_j}$
\begin{equation}
    \mathscr{I}_g(\ket{\psi_j}) =\frac{1-|\braket{\psi_j}{\psi^{*}_j}|}{2},
\end{equation}
Eq. (\ref{geo_im_lower}), is equivalent to
\begin{equation}\label{opt_mod}
    \max_{\{p_j,\ket{\psi_j}\}}\sum_{j} p_j|\braket{\psi_j}{\psi^{*}_j}| \leq \sqrt{F(\rho,\rho^T)}.
\end{equation}
Now, in the following we construct a special decomposition, which saturates the above bound. Note that any other ensemble of $\rho$ (let's say $\{q_i,\ket{\phi_i}\}$) is related to  the ensemble $\{p_j,\ket{\psi_j}\}$ in the following way 
\begin{equation}\label{basis}
    \sqrt{q_i}\ket{\phi_i} = \sum_j U^{*}_{ij}\sqrt{p_j}\ket{\psi_j},
 \end{equation}
where $U_{ji}$ are the elements of a Unitary matrix $U$. Without loss of generality, we can choose $\{p_j,\ket{\psi_j}\}$ to be the eigenvalues and eigenstates of $\rho$. 
\begin{equation} \label{star_product}
     \sqrt{q_i q_j}\braket{\phi_i}{\phi^{*}_j} = (U A U^{T})_{ij},
\end{equation}
where $A_{ij}=\sqrt{p_i p_j}\braket{\psi_i}{\psi^{*}_j} $. In \cite{maths}, the authors show that for any symmetric matrix $S$, there exists a singular value decomposition of the form
\begin{equation}
    S = Q \Sigma Q^T \label{sym},
\end{equation}
where $\Sigma$ is a positive diagonal matrix and $Q$ is a unitary matrix. Eq. (\ref{sym}) also appears in \cite{horn2012matrix} (P263, Corollary 4.4.4 (c)). The singular values of $A$, are the square roots of eigenvalues of $AA^{\dagger}(= AA^{*}$). Note that, 
\begin{eqnarray}
  (AA^{*})_{ij}  &=& \sum_kp_k\sqrt{p_i p_j}\braket{\psi_i}{\psi^{*}_k}\braket{\psi^{*}_k}{\psi_j} \\
  &=& \sqrt{p_i p_j}\bra{\psi_i}\rho^{T}\ket{\psi_j}\\
  &=&\bra{\psi_i} \sqrt{\rho} \rho^{T}\sqrt{\rho}\ket{\psi_j}.
\end{eqnarray}
Therefore, $(AA^{*})_{ij}$ (matrix) is $\sqrt{\rho}  \rho^{T}\sqrt{\rho}$ (operator), written (in matrix form) in the eigenbasis of $\rho$. This shows that the singular values of $(A)_{ij}$ (matrix) are the eigenvalues of $\sqrt{\sqrt{\rho}  \rho^{T}\sqrt{\rho}}$. Eqs. (\ref{star_product}) and (\ref{sym}) shows that one can find a decomposition (say $\{\lambda_i,\ket{\mu_i}\}$), which is ``conjugate-orthogonal'' 
\begin{equation} \label{optimal_decomposition}
   \sqrt{\lambda_i \lambda_j}\braket{\mu_i}{\mu^{*}_j} = \delta_{ij}D_{j}
\end{equation}
where, $D_j$ is the $j^{th}$ eigenvalue of $\sqrt{\sqrt{\rho}  \rho^{T}\sqrt{\rho}}$.
This completes the proof. We found out that an independent proof of this theorem can be found in \cite{Uhlmann_2000}.

\subsection{Proof of lemma \ref{equal}}\label{proof_equal}

Let $\{\lambda_j,\ket{\mu_j}\}$ be a pure state decomposition we described in Eq. (\ref{optimal_decomposition}). Therefore, $\mathscr{I}_g(\rho)=\sum_j\lambda_j\mathscr{I}_g(\ket{\mu_j})$. We now show that we can achieve a decomposition (say $\{p_i, \ket{\psi_i}\}$), where 
\begin{equation}
    \braket{\psi_i}{\psi^{*}_i}= \braket{\psi_j}{\psi^{*}_j} \,\,\forall\,\, i,j.
\end{equation}
Firstly, note that 
\begin{equation}
    \braket{\mu_i}{\mu^{*}_i} \geq 0 \,\,\forall\,\, i.
\end{equation}
We further define $\braket{\mu_i}{\mu^{*}_i}$ as the ``conjugate product'' of $\ket{\mu_i}$. Therefore, the ``average conjugate product'' ($\sum_i \lambda_i \braket{\mu_i}{\mu^{*}_i}$) of the decomposition $\{\lambda_j,\ket{\mu_j}\}$ is equal to $1-\mathscr{I}_g(\rho)$. Now, we go into another decomposition such that the ``average conjugate product'' remains constant.

As we know (from Eq.(\ref{basis})), we can go into another decomposition by mixing the sub-normalised pure states of the decomposition with elements of a unitary matrix (For our case, we just need a orthogonal matrix). The procedure to achieve the desired decomposition goes as follows.

If $\mathscr{I}_g(\ket{\mu_i})=\mathscr{I}_g(\rho)$, we already have the required decomposition. In the other case, there exist at least two pure states (without loss of generality say) $\ket{\mu_1}$ and $\ket{\mu_2}$, such that
\begin{equation}
    \mathscr{I}_g(\ket{\mu_1})>\mathscr{I}_g(\rho)\,\, \textrm{and}\,\, \mathscr{I}_g(\ket{\mu_2})<\mathscr{I}_g(\rho).
\end{equation}
Now we mix these two sub-normalised pure states in the following way to achieve a new decomposition.
\begin{equation}
\begin{split}
     \sqrt{\lambda_1'}\ket{\mu'_1} &= \cos{\alpha}\sqrt{\lambda_1}\ket{\mu_1} + \sin{\alpha}\sqrt{\lambda_{2}}\ket{\mu_2}\, \mbox{and}\\
     \sqrt{\lambda_2'}\ket{\mu'_2} &= -\sin{\alpha}\sqrt{\lambda_{1}}\ket{\mu_1} + \cos{\alpha}\sqrt{\lambda_{2}}\ket{\mu_2}.
     \end{split}
\end{equation}
 Without loss of generality, we can take $\alpha \in [0,\pi/2]$. Since, $\mathscr{I}_g(\ket{\mu'_1})$ varies continuously with $\alpha$ and it is easy to see that, at $\alpha = 0$
\begin{equation}
    \ket{\mu'_{1}} = \ket{\mu_{1}} \implies \mathscr{I}_g(\ket{\mu'_1}) = \mathscr{I}_g(\ket{\mu_1})
\end{equation}
and at $\alpha$ = $\pi/2$
\begin{equation}
    \ket{\mu'_{1}} = \ket{\mu_{2}} \implies \mathscr{I}_g(\ket{\mu'_1}) = \mathscr{I}_g(\ket{\mu_2}).
\end{equation}
Therefore, there exists a $\alpha$ = $\alpha^{*}$, such that 
\begin{eqnarray}
 &&\sqrt{\lambda'_{1}}\ket{\mu'_{1}} = \cos{(\alpha^{*})}\sqrt{\lambda_{1}}\ket{\mu_{1}} + \sin{(\alpha^{*})}\sqrt{\lambda_{2}}\ket{\mu_{2}}\, \mbox{and}\nonumber\\
 &&\mathscr{I}_g(\ket{\mu'_1}) =\mathscr{I}_g({\rho}).
\end{eqnarray}
Now after this we have a new decomposition of $\rho$. Note, that the average conjugate product of the new decomposition is same as the previous one and $\{\ket{\mu'_2},\ket{\mu_j}\},j=3,4...$ are cojugate-orthonormal. Therefore, we can continue this procedure recursively for the remaining vectors until we reach the required decomposition.

\subsection{Proof of theorem \ref{thm:Geom_Continuity}}\label{proof_geom_continuity}
 Let us make use of the following distance measure:
\begin{equation}
    D(\rho, \sigma) = \cos^{-1}\left(\sqrt{F(\rho, \sigma)}\right), \label{eq:BuresAngle}
\end{equation}
known as the Bures angle. It has the following properties: i) $D(\rho,\sigma)\geq 0$ for all $\rho$ and $\sigma$, ii) $D(\rho,\sigma)= 0$ if an only if $\rho=\sigma$, iii) It satisfies the triangle inequality, i.e., $D(\rho,\sigma)\leq D(\rho,\eta)+D(\eta,\sigma)$. Let us note that, the geometric measure of imaginarity for any state lies between 0 to $1/2$. Consider two states $\rho$ and $\rho'$ with $\rho'\in S_{\rho,f}$. Further we assume that $\rho_r$ and $\rho_r'$ are the closest real states to $\rho$ and  $\rho'$ (with respect to Bures angle) respectively. Using the fact that $\rho_r$ is the closest real state to $\rho$ and the triangle inequality, we find
\begin{eqnarray}
\cos^{-1}\left(\sqrt{1-\mathscr{I}_g(\rho)}\right) &=& D(\rho,\rho_r) \leq D(\rho,\rho'_r) \\ \nonumber
&\leq& D(\rho,\rho')+D(\rho', \rho'_r)\\ \nonumber
&\leq& \cos^{-1}(\sqrt{f}) + \cos^{-1}\left(\sqrt{1-\mathscr{I}_g(\rho')}\right).
\end{eqnarray}
Simplifying this, we have
\begin{equation} \label{geometric_lower_bound}
    \mathscr{I}_g(\rho') \geq \sin^{2}\left(\max\left\{ \sin^{-1}\!\sqrt{\mathscr{I}_g(\rho)}-\cos^{-1}\!\sqrt{f},0\right\} \right).
\end{equation}
An upper bound on $\mathscr{I}_g(\rho')$ can be derived in similar fashion. Again using the fact that $\rho'_r$ is the closest real state to $\rho'$ and the triangle inequality, we get
\begin{eqnarray}
\cos^{-1}\left(\sqrt{1-\mathscr{I}_g(\rho')}\right) &=& D(\rho',\rho'_r)\leq D(\rho',\rho_r) \\ 
&\leq& D(\rho,\rho')+D(\rho, \rho_r)\nonumber\\
&\leq& \cos^{-1}(\sqrt{f}) + \cos^{-1}\left(\sqrt{1-\mathscr{I}_g(\rho)}\right)\nonumber. 
\end{eqnarray}
Simplifying this, we will end up with the following upper bound
\begin{equation} \label{geometric_upper_bound}
    \mathscr{I}_g(\rho') \leq \sin^{2}\left(\min\left\{ \sin^{-1}\!\sqrt{\mathscr{I}_g(\rho)}+\cos^{-1}\!\sqrt{f},\frac{\pi}{4}\right\} \right).
\end{equation}
Now we show that lower bound in Eq. (\ref{geometric_lower_bound}) can be achieved. From lemma \ref{equal}, we know that for any state $\rho$, we can find a pure state decomposition such that $\rho = \sum_i p_i \ket{\psi_i}\!\bra{\psi_i}$ and $\mathscr{I}_g(\psi_i) = \mathscr{I}_g(\rho)$ holds true for all states $\ket{\psi_i}$. This implies that each of the states $\ket{\psi_i}$ can be written as
\begin{eqnarray}
    \ket{\psi_i} &= \cos{\alpha}\ket{a_i} + i \sin{\alpha}\ket{a^{\perp}_i},
\end{eqnarray}
where $\braket{a_i}{a^{\perp}_i} = 0$ and $\ket{a_i}$ and $\ket{a^{\perp}_i}$ are real states. Here, we assumed, $\alpha=\sin^{-1}\sqrt{\mathscr{I}_g(\rho)}$. Note that $\alpha \in\{0,\frac{\pi}{4}\}$. Let's consider another state $\rho_{\min}$ to be
\begin{align}\label{minimum_decomposition}
    \rho_{\min} = \sum_i q_i \ket{\phi_i}\!\bra{\phi_i}
\end{align}
with
\begin{eqnarray}
    \ket{\phi_i} &= \cos{\tilde{\beta}}\ket{a_i}  + \sin{\tilde{\beta}}\ket{a^{\perp}_i},
\end{eqnarray}
where $\tilde \beta \in [0,\pi/4]$. The probabilities are given by
\begin{equation}
    q_i = \frac{p_i |\braket{\psi_i}{\phi_i}|^2}{\sum_k p_k |\braket{\psi_k}{\phi_k}|^2}. \label{prob}
\end{equation}
From the convexity of the geometric measure of imaginarity, we obtain
\begin{equation}
 \mathscr{I}_g(\rho_{\min})\leq \sin^2\tilde{\beta}. \label{convexity}
\end{equation}
Now we use the properties of root-fidelity \cite{wilde_2017} to obtain
\begin{align}
   \sqrt{F(\rho, \rho_{\min})} &\geq \sum_i \sqrt{p_i q_i} |\braket{\psi_i}{\phi_i}| \nonumber \\
    &= \sqrt{\sum_i p_i |\braket{\psi_i}{\phi_i}|^2} = |\cos (\alpha -\tilde{\beta})|.
    \end{align}
    Therefore,
    \begin{equation}
        F(\rho, \rho_{\min}) \geq \cos^{2}(\alpha -\tilde{\beta}). \label{fidelity}
    \end{equation}
Now, let us assume $\cos^{-1}\sqrt{f}=k$ and set
\begin{equation}
    \tilde{\beta}=\max\{\alpha-k,0\}. \label{eq:btilde}
\end{equation}
For this choice of $\tilde{\beta}$, we get
    \begin{equation}\label{ineq1}
        F(\rho, \rho_{\min}) \geq \cos^{2}\left(\min\{k,\alpha\}\right) \geq \cos^{2}k\geq f. \end{equation}
Here, we used the fact that $\cos^2$ is a monotonically decreasing function in $[0, \pi/4]$ and $\min\{k,\alpha\}$ is within $[0, \pi/4]$. This shows that, the state $\rho_{\min}$ is within $S_{\rho,f}$. Further, from Eq.~(\ref{convexity}) we get
\begin{eqnarray}
\mathscr{I}_g(\rho_{\min})&\leq&\sin^{2}\left(\max\{\alpha-k,0\}\right)\nonumber\\
&\leq&\sin^{2}\left(\max\left\{ \sin^{-1}\!\sqrt{\mathscr{I}_g(\rho)}-\cos^{-1}\!\sqrt{f},0\right\} \right). \label{eq:btildeUpperBound}
\end{eqnarray}
Comparing Eqs. (\ref{geometric_lower_bound})and (\ref{eq:btildeUpperBound}), we obtain 
\begin{equation}
 \mathscr{I}_g(\rho_{\min}) =  \sin^{2}\left(\max\left\{ \sin^{-1}\!\sqrt{\mathscr{I}_g(\rho)}-\cos^{-1}\!\sqrt{f},0\right\} \right).
 \end{equation}
Hence, the minimum geometric measure of imaginarity within the set $S_{\rho,f}$ is given by 
\begin{equation}\label{min_geo}
\min_{\rho'\in S_{\rho,f}}\mathscr{I}_g(\rho') 
=\sin^{2}\left(\max\left\{ \sin^{-1}\!\sqrt{\mathscr{I}_g(\rho)}-\cos^{-1}\!\sqrt{f},0\right\} \right).
\end{equation}
Now, we focus on the case when $\rho$ is a pure state (let's say $\ket{\psi}$):
\begin{eqnarray}
\ket{\psi}=\cos{\alpha}\ket{a} + i \sin{\alpha}\ket{a^{\perp}}.
\end{eqnarray}
For this case, the upper bound in Eq.~(\ref{geometric_upper_bound}) is achievable as well. To show this, we consider the following pure state
\begin{align}
    \ket{\psi_{\max}} &= \cos(\min \{ \alpha + k, \pi/4\})\ket{a} \label{pure_max}\\
    &+ i \sin(\min \{ \alpha + k, \pi/4\})\ket{a^{\perp}}. \nonumber
\end{align}
Note that 
\begin{align} 
\mathscr{I}_g(\ket{\psi_{\max}}) &= \sin^2(\min \{ \alpha + k, \pi/4\}). 
\end{align}
We also note that
\begin{equation}
F(\psi, \psi_{\max}) = \cos^{2}\left(\min \left\{k,\frac{\pi}{4}-\alpha\right\} \right) \geq \cos^{2}k\geq f. 
\end{equation}
Here we use the fact that, $\cos^2$ is a monotonically decreasing function in $[0, \pi/4]$ and $\min \{k,\pi/4-\alpha\}$ is within $[0, \pi/4]$. Therefore, $\ket{\psi_{\max}}$ is inside $S_{\psi,f}$ and has the maximum possible geometric measure of imaginarity given by the following equation:
\begin{eqnarray}
\max_{\rho'\in S_{\psi,f}}\mathscr{I}_g(\rho') &=& \sin^2\left(\min\{\alpha+k,\pi/4\}\right)\\
&=&\sin^{2}\left(\min\left\{ \sin^{-1}\!\sqrt{\mathscr{I}_g(\psi)}+\cos^{-1}\!\sqrt{f},\pi/4\right\} \right).\nonumber
\end{eqnarray}
This completes the proof.

\subsection{Proof of theorem \ref{thm:PureConversion}}
From, Eq. (\ref{opt_prob}), we know that optimal probability to achieve $\rho$ from a pure state $\ket{\psi}$ with unit fidelity is given by
\begin{equation}
    P(\ket{\psi} \rightarrow \rho) =\min \left\{\frac{\mathscr{I}_g(\ket{\psi})}{\mathscr{I}_g(\rho)},1\right\}. 
\end{equation}
If we want to achieve $\rho$, with fidelity $f$ at least, the optimal strategy is to go to a state ($\rho'$) with a minimal geometric measure of imaginarity in $S_{\rho,f}$. Therefore,
\begin{equation}
P_{f}(\ket{\psi}\rightarrow\rho)=\min\left\{\frac{\mathscr{I}_g(\ket{\psi})}{\min_{\rho'\in S_{\rho,f}}\mathscr{I}_g(\rho')},1\right\}.
\end{equation}
From Eq. (\ref{min_geo}), we know that
\begin{align}
\min_{\rho'\in S_{\rho,f}}\mathscr{I}_g(\rho') & =\sin^{2}\left(\max\left\{ \sin^{-1}\!\sqrt{\mathscr{I}_g(\rho)}-\cos^{-1}\!\sqrt{f},0\right\} \right).
\end{align}
Let us now define,
\begin{equation}
    m_1=\sin^{-1}\sqrt{\mathscr{I}_g(\ket{\psi})}-\sin^{-1}\sqrt{\mathscr{I}_g(\rho)} +\cos^{-1}\sqrt{f},
\end{equation}
considering the case when $m_1\geq0$, we get
\begin{equation}
    \sin^{-1}\sqrt{\mathscr{I}_g(\rho)}-\cos^{-1}\sqrt{f} \leq \sin^{-1}\sqrt{\mathscr{I}_g(\ket{\psi})}.
\end{equation}
We know that
\begin{equation}
\sin^{-1}\sqrt{\mathscr{I}_g(\rho)}-\cos^{-1}\sqrt{f}\in[-\pi/2,\pi/4]
\end{equation}
and $\sin^{-1}\sqrt{\mathscr{I}_g(\ket{\psi})}\in[0,\pi/4]$. Therefore, 
\begin{equation}
    \max\left\{\sin^{-1}\sqrt{\mathscr{I}_g(\rho)}-\cos^{-1}\sqrt{f},0\right\}\leq \sin^{-1}\sqrt{\mathscr{I}_g(\ket{\psi})}.
\end{equation}
Using these, we get
\begin{align}
\min_{\rho'\in S_{\rho,f}}\mathscr{I}_g(\rho') & =\sin^{2}\left(\max\left\{ \sin^{-1}\sqrt{\mathscr{I}_g(\rho)}-\cos^{-1}\sqrt{f}),0\right\} \right)\\
 & \leq\sin^{2}\left(\sin^{-1}\sqrt{\mathscr{I}_g(\ket{\psi})}\right)=\mathscr{I}_g(\ket{\psi}). \nonumber
\end{align}
When $\mathscr{I}_g(\ket{\psi}) > 0$,
\begin{equation}
    \frac{\mathscr{I}_g(\ket{\psi})}{\min_{\rho'\in S_{\rho,f}}\mathscr{I}_g(\rho')} \geq 1.
\end{equation}
This shows that $P_{f}(\ket{\psi}\rightarrow\rho)=1$ when $m_1\geq 0 $. 

Now, we look at the other case when, $m_1<0$, expressed as
\begin{align}
   \sin^{-1}\sqrt{\mathscr{I}_g(\rho)}-\cos^{-1}\sqrt{f}>\sin^{-1}\sqrt{\mathscr{I}_g(\ket{\psi})}>0.
\end{align}
From Eq. (\ref{opt_prob}), we have
\begin{equation}\label{eq:optimal_probability_imaginarity}
    P_{f}(\ket{\psi}\rightarrow\rho)=\frac{\mathscr{I}_g(\ket{\psi})}{\sin^{2}(\sin^{-1}\sqrt{\mathscr{I}_g(\rho)}-\cos^{-1}\sqrt{f}) }.
\end{equation}
Using this result, one can also find a closed expression for $F_p$. Let's first consider the case when $p \leq \frac{\mathscr{I}_g(\psi)}{\mathscr{I}_g(\rho)}<1$, one can see that, in this case $F_p (\psi \rightarrow \rho) = 1$. Let us remind that, if $G(\psi)\geq G(\rho)$, then the transformation is always possible exactly with unit probability. When $1\geq p>\frac{\mathscr{I}_g(\psi)}{\mathscr{I}_g(\rho)}$, the optimal achievable fidelity can be obtained by solving Eq.~(\ref{eq:optimal_probability_imaginarity}) for $f$, which gives 
\begin{eqnarray}
F_p (\psi \rightarrow \rho)=\cos^{2}\left[\sin^{-1}\!\sqrt{\mathscr{I}_g(\rho)}-\sin^{-1}\!\sqrt{\frac{\mathscr{I}_g(\psi)}{p}}\right].
\end{eqnarray}
This completes the proof.
\bibliography{imaginarity}
\end{document}